\documentclass[a4paper]{article}

\usepackage{fullpage}
\usepackage{amsthm,amsmath,amssymb}
\usepackage{algorithm, algorithmic}
\usepackage{mdframed}
\usepackage{graphicx}
\usepackage{enumitem}
\usepackage{xspace}

\newenvironment{properties}[2][1]{
	\begin{enumerate}[label={\normalfont (\ref{#2}\alph*)},labelwidth=*,leftmargin=*,itemsep=0pt,topsep=3pt, start=#1]
	}{
\end{enumerate}
}

\newtheorem{theorem}{Theorem}[section]
\newtheorem{lemma}[theorem]{Lemma}

\newtheorem{corollary}[theorem]{Corollary}
\newtheorem{observation}[theorem]{Observation}

\newcommand{\set}[1]{\left\{#1\right\}}

\DeclareMathOperator{\union}{\bigcup}

\renewcommand{\hat}{\widehat}
\renewcommand{\tilde}{\widetilde}

\newcommand{\R}{\mathbb{R}}
\newcommand{\Z}{\mathbb{Z}}

\newcommand{\calS}{{\mathcal S}}

\newcommand{\laminarKC}{\textsf{Laminar-KC}\xspace}
\newcommand{\intervalKC}{\textsf{Interval-KC}\xspace}
\newcommand{\lotsizing}{\textsf{CMILS}\xspace}
\newcommand{\KC}{\textsf{KC}\xspace}

\newcommand{\hcost}{\mathsf{hcost}}

\newcommand{\newi}{{i}}
 
\title{Constant Approximation Algorithm for Non-Uniform Capacitated Multi-Item Lot-Sizing via Strong Covering Inequalities}

\author{Shi Li\thanks{Department of Computer Science and Engineering, University at Buffalo, 1 White Road, Buffalo, NY 14260.  {\tt shil@buffalo.edu}. Supported in part by NSF grant CCF-1566356.}}

\date{}

\begin{document}
	\maketitle
	\begin{abstract}
	We study the non-uniform capacitated multi-item lot-sizing (\lotsizing) problem. In this problem, there is a set of demands over a planning horizon of $T$ time periods and all demands must be satisfied on time.  We can place an order at the beginning of each period $s$, incurring an ordering cost $K_s$. The total quantity of all products ordered at time $s$ can not exceed a given capacity $C_s$. On the other hand, carrying inventory from time to time incurs inventory holding cost. The goal of the problem is to find a feasible solution that minimizes the sum of ordering and holding costs.
	
	Levi et al.\ (Levi, Lodi and Sviridenko, Mathmatics of Operations Research 33(2), 2008) gave a 2-approximation for the problem when the capacities $C_s$ are the same.  In this paper, we extend their result to the case of non-uniform capacities. That is, we give a constant approximation algorithm for the capacitated multi-item lot-sizing problem with  general capacities.
	
	The constant approximation is achieved  by adding an exponentially large set of new covering inequalities to the natural facility-location type linear programming relaxation for the problem.  Along the way of our algorithm, we reduce the \lotsizing problem to two generalizations of the classic knapsack covering problem. We give LP-based constant approximation algorithms for both generalizations, via the iterative rounding technique. 
\end{abstract}
\thispagestyle{empty}
\newpage

\setcounter{page}{1}

\section{Introduction}

Since the seminal papers of Wagner and Whitin \cite{WW58} and Manne \cite{Man58} in late 1950's,  lot-sizing problems has become one of the most important classes of problems in inventory management and production planning (\cite{PW06}). Given a sequence of time-varying demands for different products over a time horizon, a lot-sizing problem asks for the time periods for which productions and orders take place and the quantities of products to be produced and ordered, so as to minimize the total production, ordering and inventory holding cost.  

In many practical settings, due to the shortage of resources such as manpower, equipments and budget, there are (possibly time-dependent) capacity constraints on the total units of products that can be produced or ordered at a time. Thus, when modeling the lot-sizing problems, it is important to take these capacity constraints into account.  Often, the capacity constraints make the problems computationally harder. There has been an immense amount of work on capacitated lot-sizing problems, from the perspective of integer programming, heuristics, as well as tractability of special cases.

In this paper, we study the single-level capacitated multi-item lot-sizing (\lotsizing) problem, where good decisions have to be made to balance two costs: ordering cost and inventory holding cost.  Placing an order at some time $s$ incurs a time-dependent ordering cost $K_s$, making it too costly to place an order at every time period.  On the other hand, placing a few orders will result in holding products in inventory to satisfy future demands, incurring high holding cost.   When an order is placed at time $s$, there is a capacity $C_s$ on the total amount of products that can be ordered.  

We study the \lotsizing problem from the perspective of approximation algorithms. When the capacities are uniform, i.e, the capacities $C_s$ for different time periods $s$ are the same, Levi et al.\ \cite{LLS08} gave a $2$-approximation algorithm for the problem. They used the flow covering inequalities that were introduced in \cite{PRW85}, to reduce the unbounded integrality gap of the natural facility-location type LP relaxation for the \lotsizing problem to $2$.  However, the flow-covering inequalities heavily used the uniform-capacity property; it seems hard to extend them to the problem with non-uniform capacities.

In this paper, we give a 10-approximation algorithm for non-uniform capacitated multi-item lot-sizing problem.  Inspired by the ``effective capacity'' idea of the knapsack covering inequalities introduced by Carr et al.\ \cite{CFL00}, we introduce a set of covering inequalities that strengthen the natural facility-location type LP relaxation. We believe our covering inequalities can be applied to many other problems with non-uniform capacities. In the cutting-plane method for solving the integer programmings for capacitated problems, our covering inequalities may be used to generate an initial solution that is provably good.   Along the way of our algorithm, we reduce the \lotsizing problem to two generalizations of the classic knapsack covering problem, namely, the interval and laminar knapsack covering problems. We give an iterative rounding algorithm for the laminar knapsack covering problem.  The two generalizations and the use of iterative rounding may be of independent interest. 

\paragraph{Problem Definition} 

In the capacitated multi-item lot-sizing (\lotsizing) problem, there is a finite time horizon of $T$ discrete time periods indexed by $[T]$, and a set of $N$ items indexed by $[N]$.  For each item $i \in {[N]}$, we have a demand that $d_i > 0$ units of item $i$ must be ordered by time $r_i \in [T]$.  For each $s \in [T]$, we can place an order at the beginning of period $s$,  incurring an \emph{ordering cost} $K_s > 0$.  If an order is placed at time $s \in [T]$, we can order any subset of items. However, the quantity of total units ordered for all items can not exceed a given capacity $C_s > 0$.  Carrying inventory over periods incurs \emph{holding costs}. For every $i \in {[N]}, s \in [r_i]$, we use $h_{i}(s)$ to denote the per-unit cost of holding one unit of item $i$ from period $s$ to period $r_i$.  We assume that $h_{i}$ is non-increasing. Also, a unit of item $i$ ordered at the beginning of period $r_i$ can be used to satisfy the demand for item $i$ immediately; thus we assume $h_{i}(r_i) = 0$. The goal of the problem is to satisfy all the demands so as to minimize the sum of ordering and holding costs.  

The mixed integer programming for the problem is given in \ref{MIP:lotsizing}, 	which is a facility-location type programming. For every $i \in {[N]}$ and $s \in [r_i]$, $x_{s, i} \in [0, 1]$ specifies the fraction of demand $i$ that is ordered at time $s$, and for every $s \in [T]$, $y_s \in \{0, 1\}$ indicates whether we place an order at time $s$ or not.  The goal of the MIP is to minimize the sum of the ordering cost $\sum_{s \in [T]}y_s K_s$ and the holding cost $\sum_{i \in {[N]}}d_i\sum_{s \in [r_i]}x_{s,i}h_i(s)$.  Constraint~\eqref{LPC:ls-items-covered} requires that the demand for every item $i$ is fully satisfied, Constraint~\eqref{LPC:ls-x-leq-y} restricts that we can order an item at time $s$ only if we placed an order at time $s$, and Constraint~\eqref{LPC:ls-capacity} requires that the total amount of demand ordered at time $s$ does not exceed $y_sC_s$. Notice that in a feasible solution $(x^*, y^*)$ to the \lotsizing problem, $y^*$ has to be integral, but $x^*$ can be fractional. 

\begin{figure*}[h]\begin{mdframed}
	\begin{equation}
	\min \quad \sum_{s\in [T]} y_s K_s +  \sum_{i \in {[N]}}d_i \sum_{s \in [r_i]} x_{s, i} h_i(s) \qquad \text{s.t.}  \tag{$\mathsf{MIP}_\mathsf{CMILS}$}  \label{MIP:lotsizing}
	\end{equation} \vspace*{-1.2\abovedisplayskip}\vspace*{-1.2\belowdisplayskip}
	
	\noindent\begin{tabular}{p{0.45\textwidth}p{0.52\textwidth}}
		{\begin{alignat}{2}
			\sum_{s \in [r_i]}x_{s, i} &= 1, &\qquad \forall i &\in {[N]}; \label{LPC:ls-items-covered} \\
			x_{s, i} &\leq y_s, &\qquad \forall i &\in {[N]}, s \in [r_i];  \label{LPC:ls-x-leq-y}
		\end{alignat}} &
		{\begin{alignat}{2} 
			\sum_{i \in {[N]}: r_i \geq s}x_{s, i}d_i &\leq y_s C_s, &\qquad \forall s &\in [T]; \label{LPC:ls-capacity} \\
			x_{s, i} &\in [0, 1], &\qquad \forall i &\in {[N]}, s \in [r_i];  \label{LPC:ls-between-01} \\
			y_{s} &\in \{0, 1\}, & \qquad \forall s &\in [T]. \label{LPC:integral}
		\end{alignat}}
	\end{tabular} \vspace*{-18pt}
\end{mdframed}\end{figure*}
	
In the traditional setting for multi-item lot-sizing, there is a demand of value $d_{s, i}$ for each item $i$ at each time period $s$, and the holding cost is linear for each item $i$: there will be a per-unit host cost $h'_i(s)$ for carrying one unit of item $i$ from period $s$ to $s+1$.  The way we defined our holding cost functions allows us to assume that there is only one demand for each item: if there are multiple demands for an item at different periods, we can think of that the demands are for different items, by setting the holding cost functions for these items correctly.  This assumption simplifies our notation: instead of using an $(s, i)$-pair to denote a demand, we can use a single index $i$. Thus, we do not distinguish between demands and items: an element $i \in {[N]}$ is both a demand and an item.   In our setting, the traditional single-item lot-sizing problem corresponds to the case where there is a function $h': [T-1]\to \R_{\geq 0}$ such that $h_i(s) = \sum_{t = s}^{r_i-1}h'(t)$ for every $i \in [N]$ and $s \in [r_i]$. 

\paragraph{Know Results}  For the single-item lot-sizing problem, the seminal paper of Wagner and Whitin \cite{WW58} gave an efficient dynamic programming algorithm for the uncapacitated version. Later, efficient DP algorithms were also found for the single-item lot-sizing problem with uniform hard capacities (\cite{FLR80}) and uniform soft capacities (\cite{PW93}). When the capacities are non-uniform, the problem becomes weakly NP-hard (\cite{FLR80}), but admits an FPTAS (\cite{HW01}). Thus, we have a complete understanding of the status of the single-item lot-sizing problem.

For the multi-item lot-sizing problem, the dynamic programming in \cite{WW58} carries over if the problem is uncapacitated. Levi et al.\ \cite{LLS08} showed that the uniform capacitated multi-item lot-sizing problem is already strongly NP-hard, and gave a $2$-approximation algorithm for the problem. Special cases of the \lotsizing problem have been studied in the literature.  Anily and Tzur \cite{AT05} gave a dynamic programming algorithm for the special case where the capacities and ordering costs are both uniform, and the number of items is a constant. (In our model, this means that there is a family of $O(1)$ functions $h':[T-1] \to R_{\geq 0}$, such that for every $i \in [N]$, there is a function $h'$ from the family such that $h_i(s) = \sum_{t = s}^{r_i-1}h'(t)$  for every $s \in [r_i]$.) Anily et al.\ \cite{ATW09} considered the \lotsizing problem with uniform capacities and the monotonicity assumption for the holding cost functions. This assumption says that there is an ordering of items according to their importance.  For every time period $s$, the cost of holding an unit of item $i$ from period $s$ to $s+1$, is higher than (or equal to) the cost of holding one unit of a less important item from period $s$ to $s+1$. \cite{ATW09} gave a small size linear programming formulation to solve this special case exactly.  Later, Even et al.\ \cite{ELR08} gave a dynamic programming for the problem when demands are polynomially bounded and the holding cost functions satisfy the monotonicity assumption. The assumption of polynomially bounded demands is necessary even with the monotonicity assumption, since otherwise the problem is weakly NP-hard (\cite{FLR80}). 

\subsection{Our Reults and Techniques}
Our main result is a constant approximation algorithm for the capacitated multi-item lot-sizing problem. 
\begin{theorem}\label{thm:lot-sizing}
	There is a $10$-approximation algorithm for the multi-item lot-sizing problem with non-uniform capacities. 
\end{theorem}

We give an overview of our techniques used for proving Theorem~\ref{thm:lot-sizing}.  We start with the natural facility-location type linear programming relaxation for the \lotsizing problem, which is obtained from \ref{MIP:lotsizing} by relaxing Constraint~\eqref{LPC:integral} to $y_s \in [0, 1], \forall s \in [T]$.  Since Constraint~\eqref{LPC:ls-items-covered} requires $\sum_{s \in [r_i]} x_{s, i} = 1$ for every $i \in [N]$, $\{x_{s, i}\}_{s \in [r_i]}$ form a distribution over $[r_i]$. We define the tail of the distribution to be the set of some latest periods in $[r_i]$ that contributes a constant mass to the distribution. 

For simplicity, let us first assume that the capacities for the orders are soft capacities. That is, we are allowed to place multiple orders at each period $s$ in our final solution, by paying a cost $K_s$ for each order.  In this case, the integrality gap of the natural LP relaxation is $O(1)$.  In our rounding algorithm, we scale up the $y$ variables by a large constant. With the scaling, we require that each item $i \in [N]$ is only ordered in the tail of the distribution for $i$. In this way, the holding cost can be bounded by a constant times the holding cost of the LP solution.  

With these requirements, the remaining problem becomes an ``interval knapsack covering'' (\intervalKC) problem, where we view each period $s \in [T]$ as a knapsack of capacity $C_s$ and cost $K_s$. If we place an order at period $s$, then we select the knapsack $s$. For every $a, b \in \Z$ such that $0 \leq a < b < T$, we require that the total capacity of selected knapsacks in $(a, b]$ is at least $R_{a, b}$.  The goal of the \intervalKC problem is to select a set of knapsacks of the minimum cost that satisfies all the requirements.  We give a constant approximation algorithm for \intervalKC based on iterative rounding. Specifically, we first reduce the \intervalKC instance to an instance of a more restricted problem, called ``laminar knapsack covering'' (\laminarKC) problem, in which the intervals $(a, b]$ with positive $R_{a,b}$ form a laminar family. The laminar structure allows us to apply the iterative rounding technique to solve the problem: we maintain an LP and in each iteration the LP is solved to obtain a vertex-point solution; as the algorithm proceeds, more and more variables will become integral and finally the algorithm will terminate.  Overall, the LP solution to the \lotsizing instance leads to a good LP solution to the \intervalKC instance, which in turn leads to a good LP solution to the \laminarKC instance. Therefore, the final ordering cost can be bounded.

However, when the capacities are hard capacities, the integrality gap of the natural LP becomes unbounded. The issue with the above algorithm is that some $y_s$ might have value more than a constant, and scaling it up will make it more than $1$.  So our final solution needs to include many orders at period $s$. Indeed, the same issue occurs for many other problems such as knapsack covering and capacitated facility location, for which the linear programming has $O(1)$ integrality gap  for the soft capacitated version but  unbounded gap for the hard capacitated version. For both problems, stronger LP relaxations are known to overcome the gap instances (\cite{ASS14, CFL00}).  
Our idea behind the strengthened LP is similar to those in \cite{CFL00} and \cite{ASS14}.  If for some $s \in [T]$, $y_s$ is larger than a constant, then we can afford to place an order at time $s$.   So, we break $[T]$ into two sets: $S^+$ contains the periods in which we already placed an order, and $[T] \setminus S^+$ contains the periods $s$ with small $y_s$ value. We consider the residual instance obtained from the input \lotsizing instance by pre-selecting the orders in $S^+$. In this residual instance, all periods $s$ other than the ones with pre-selected orders have small $y_s$ values; thus we hope to run the algorithm for the soft-capacitated multi-item lot-sizing problem.  The challenge is that we can not remove $S^+$ from the residual instance since we do not know what items are assigned to the orders in $S^+$.  To overcome this issue, we introduce a set of strong covering inequalities. These inequalities use the ``effective capacity'' idea from the knapsack covering inequalities introduced in \cite{CFL00}: if we want to use some orders to satisfy $d$ units of demand, then the ``effective capacity'' of an order at time $s$ is $\min\set{C_s, d}$, instead of $C_s$. 
 	\section{Preliminaries}

\paragraph{Notations}
Throughout this paper, $C$ and $K$ will always be vectors in $\R_{\geq 0}^T$.  For any set $S \subseteq [T]$, we define $C(S):= \sum_{s \in S} C_s$ and $K(S):=\sum_{s \in S}K_s$.  For every set $I \subseteq [N]$ of items, we use $d(I) := \sum_{i \in I} d_i$ to denote the total demand for items in $I$.  Since we shall use intervals of integers frequently, we use $(a, b), (a, b], [a, b), [a, b]$ to denote sets of integers in these intervals; the only exception is that $[0, 1]$ will still denote the set of real numbers between $0$ and $1$ (including $0$ and $1$). In particular, an interval over $[T]$ is some $(a, b]$, for some integers $a, b$ such that $0 \leq a < b \leq T$.

\paragraph{Knapsack Covering Inequalities}

In the classic knapsack covering (\KC) problem, we are given a set $[T]$ of knapsacks, where each knapsack $s \in [T]$ has a capacity $C_s > 0$ and a cost $K_s \geq 0$. The goal of the problem is to select a subset $S^* \subseteq [T]$ with $C(S^*) \geq R$, so as to minimize $K(S^*)$. That is, we want to select a set of knapsacks with total capacity at least $R$ so as to minimize the total cost of selected knapsacks. In the \lotsizing problem, if we only have one item $1$ with demand $d_1 = R$ and $r_1 = T$ and the holding cost is identically 0, then the problem is reduced to the \KC problem. 

It is well-known that the \KC problem is weakly NP-hard and admits an FPTAS based on dynamic programming. However, for many problems that involve capacity covering constraints, it is hard to incorporate the dynamic programming technique. This motivates the study of LP relaxations for \KC. The naive LP relaxation, which is $\min \sum_{s \in [T]}y_sK_s$ subject to $\sum_{s \in [T]}y_sC_s \geq R$ and $y \in [0, 1]^{T}$,  has unbounded integrality gap, even if we assume that $C_s \leq R$ for every $s \in [T]$. Consider the instance with $T=2, C_1 = R-1, C_2 = R, K_1 = 0$ and $K_2 = 1$, where $R$ is a very large number.  The LP solution $(y_1, y_2) = (1, 1/R)$ has cost $1/R$; however, the optimum solution to the problem has cost $1$. 

To overcome the above gap instance, Carr et al.\ \cite{CFL00} introduced a set of valid inequalities, that are satisfied by all integral solutions to the given \KC instance. The inequalities are called \emph{knapsack covering (\KC) inequalities}, which are defined as follows:
\begin{equation*}
	\sum_{s \in [T] \setminus S} \min\set{C_s, R - C(S)} y_s \geq R - C(S), \quad \forall S \subseteq [T] \text{ s.t. } C(S) < R. \tag{Knapsack Covering Inequality}
\end{equation*}
For each $S \subseteq [T]$ with $C(S) < R$, the above inequality requires the selected knapsacks in $[T] \setminus S$ to have a total capacity at least $R-C(S)$. In this residual problem, if the capacity of a knapsack $s \in [T] \setminus S$ is more than $R-C(S)$, its ``effective capacity'' is only $R-C(S)$. So the \KC inequalities are valid.  For the above gap instance, the \KC inequality for $S = \{1\}$ requires $\min\set{C_2, R - (R-1)}y_2 \geq R - (R-1)$, i.e, $y_2 \geq 1$.  Thus the KC inequalities can handle the gap instance. Indeed, it is shown in \cite{CFL00} that the LP with all the \KC inequalities has integrality gap $2$.

\paragraph{Rounding Procedure That Detects Violated Inequalities} Although the \KC inequalities can strengthen the LP for \KC, the LP with all these inequalities can not be solved efficiently.  This issue can be circumvented by using a rounding algorithm that detects violated inequalities. Given a vector $y \in [0, 1]^T$, there is a rounding algorithm that either returns a feasible solution to the given \KC instance, whose cost is at most $2\sum_{s \in [T]}y_s K_s$, or returns a knapsack covering inequality that is violated by $y$. This rounding algorithm can be used as a separation oracle for the Ellipsoid method.  We keep on running the Ellipsoid method, as long as the rounding algorithm is returning a violated constraint. When the algorithm fails to give a violated constraint, it returns a feasible solution whose cost is guaranteed to be at most twice the cost of the optimum LP solution. This technique has been used in many previous results, e.g. \cite{ASS14, CFL00, LLS08, Li15, Li16}.  In particular, \cite{LLS08} used this technique to obtain their 2-approximation for the uniform-capacitated multi-item lot-sizing problem.  In this paper, we also apply the technique to obtain our 10-approximation for the problem with non-uniform capacities.

\paragraph{Laminar and Interval Knapsack Covering}  Along the way of our algorithm, we introduce two generalizations of the knapsack covering problem. In the laminar knapsack covering (\laminarKC) problem, we are given a set $[T]$ of knapsacks, each $s \in [T]$ with a capacity $C_s > 0$ and a cost $K_s \geq 0$. We are also given a laminar family $\calS$ of intervals of $[T]$. That is, for every two distinct intervals $(a, b], (a', b'] \in \calS$, we have either $(a, b] \cap (a', b'] = \emptyset$, or $(a, b] \subsetneq (a', b']$, or $(a', b'] \subsetneq (a, b]$.   For each $(a, b] \in \calS$, we are given a requirement $R_{a, b} > 0$.  The goal of the problem is to find a set $S^* \subseteq [T]$ such that $C(S^* \cap (a, b]) \geq R_{a, b}$ for every $(a, b] \in \calS$, so as to minimize $K(S^*)$.  W.l.o.g we assume that $[T] \in \calS$. We use $(T, C, K, \calS, R)$ to denote a \laminarKC instance.

The interval knapsack covering (\intervalKC) problem is a more general problem, in which we are still given a set $[T]$ of knapsacks, each $s \in [T]$ with a capacity $C_s > 0$ and a cost $K_s \geq 0$. However now we have a requirement $R_{a, b} \geq 0$ for every interval $(a, b]$ over $T$. The goal of the problem is to select a set $S^*$ of knapsacks with the minimum cost, such that $C(S^* \cap (a, b]) \geq R_{a, b}$, for every interval $(a, b]$ over $[T]$. We use $(T, C, K, R)$ to denote an \intervalKC instance.

We remark that the dynamic programming for \KC can be easily extended to give an FPTAS for the \laminarKC problem. However, we do not know how it can be applied to the more general \intervalKC problem, as well as the \lotsizing problem. Our algorithm for \lotsizing heavily uses the LP technique. We reduce the \lotsizing problem to the \intervalKC problem, which will be further reduced to the \laminarKC problem. Both reductions require the use of linear programming.  After the reductions, we obtain a fractional solution to some LP relaxation for the \laminarKC and round it to an integral solution.  We now state the two main theorems for the two problems, that will be used in our algorithm for \lotsizing. 

\begin{theorem}[\bf Main Theorem for \laminarKC] \label{thm:laminar-main}
	Let $(T, C, K, \calS, R)$ be a \laminarKC instance and $y \in [0, 1]^T$. Let $S^+ = \set{s \in {[T]}: y_s = 1}$ and $\tilde R_{a, b} := \max\set{R_{a, b}- C(S^+ \cap (a, b]), 0}$ for every $(a, b] \in \calS$. Suppose for every $(a, b] \in \calS$ with $\tilde R_{a, b}> 0$, we have
	\begin{flalign}
		\text{either} && \sum_{s \in (a, b]  \setminus S^+} \min\left\{C_s, \tilde R_{a, b}\right\} y_s   &\geq 2 \tilde R_{a, b}, && \label{inequ:laminar-main-requirement-1} \\
		\text{or} && \sum_{s \in (a, b] \setminus S^+: C_s \geq \tilde R_{a, b} } y_s &\geq 1. && \label{inequ:laminar-main-requirement-2}
	\end{flalign}
	Then, we can efficiently find a feasible solution $S^* \supseteq S^+$ to the instance such that $K(S^*) \leq \sum_{s \in {[T]}}y_s K_s$.
\end{theorem}

\begin{theorem}[\bf Main Theorem for \intervalKC] \label{thm:interval-main}
	Let $(T, C, K,  R)$ be an \intervalKC instance and $y \in [0, 1]^T$. Let $S^+ = \set{s \in {[T]}: y_s = 1}$ and $\tilde R_{a, b} := \max\set{R_{a, b}- C(S^+ \cap (a, b]), 0}$ for every interval $(a, b]$ over $[T]$. Suppose for every interval $(a, b]$ with $\tilde R_{a, b}> 0$, we have	
	\begin{flalign}
		\text{either} && \sum_{s \in (a, b] \setminus S^+} \min\set{C_s, \tilde R_{a, b}}y_s   &\geq 10 \tilde R_{a, b}, && \label{inequ:interval-main-requirement-1} \\
		\text{or} && \sum_{s \in (a, b] \setminus S^+: C_s \geq \tilde R_{a, b}} y_s &\geq 6. && \label{inequ:interval-main-requirement-2}
	\end{flalign}
	Then, we can efficiently find a feasible solution $S^* \supseteq S^+$ to the instance such that $K(S^*) \leq \sum_{s \in [T]}y_s K_s$.
\end{theorem}

In both theorems, we are given a fractional vector $y \in [0, 1]^T$ to the given instance. $S^+$ is the set of knapsacks with $y$ values being $1$ and thus can be included in our final solution $S^*$.  $\tilde R_{a, b}$ can be viewed as the residual requirement for the interval $(a, b]$ after we selected knapsacks in $S^+$.  Notice that Inequality~\eqref{inequ:laminar-main-requirement-1} (resp. Inequality~\eqref{inequ:interval-main-requirement-1}) is the knapsack covering inequality with ground set $(a, b]$, $S_1 = (a, b] \cap S^+$, and the right side multiplied by $2$ (resp. $10$).     

With Theorem~\ref{thm:interval-main}, we can give an LP-based $10$-approximation for the \intervalKC problem.  The rounding algorithm takes a vector $y \in [0, 1]^T$ as input.  We let $\hat y_s = \min\set{10 y_s, 1}$ for every $s \in [T]$, $S^+ = \set{s \in [T]:\hat y_s = 1}$ and $\tilde R_{a,b} = \max\set{R_{a,b} - C((a, b] \cap S^+), 0}$ for every interval $(a, b]$. Focus on every interval $(a, b]$ with $\tilde R_{a, b }>0$. If the \KC inequality is satisfied for the ground set $(a, b]$ and $S_1 = S^+ \cap (a, b]$, then Inequality~\eqref{inequ:interval-main-requirement-1} with $y$ replaced by $\hat y$ holds for this $(a, b]$; otherwise, we return this \KC inequality. Then, we can apply Theorem~\ref{thm:interval-main} to obtain a feasible solution $S^*$ whose cost is at most $\sum_{s \in [T]} \hat y_sK_s \leq 10\sum_{s \in [T]}y_s K_s $; this gives us a 10-approximation algorithm for the \intervalKC problem.

\begin{corollary}
	There is a 10-approximation algorithm for the \intervalKC problem. 
\end{corollary}

\paragraph{Organization} The remaining part of the paper is organized as follows. In Section~\ref{sec:lotsizing}, we give our algorithm for the \lotsizing problem, using Theorem~\ref{thm:interval-main} as a black box. In Section~\ref{sec:laminarKC}, we describe our iterative rounding algorithm for \laminarKC, which proves Theorem~\ref{thm:laminar-main}. In Appendix~\ref{sec:intervalKC}, we use Theorem~\ref{thm:laminar-main} to prove Theorem~\ref{thm:interval-main}, by reducing the given \intervalKC instance to a \laminarKC instance. 	\section{Approximation Algorithm for Capacitated Multi-Item Lot-Sizing}
\label{sec:lotsizing}	
	
In this section, we describe our 10-approximation algorithm for the \lotsizing problem, using Theorem~\ref{thm:interval-main} as a black box.  We first present our LP relaxation with strong covering inequalities. Then we define an \intervalKC instance, where knapsacks correspond to orders in the \lotsizing problem. Any feasible solution to \intervalKC instance gives a set of orders, for which there is a way to satisfy the demands with small holding cost. Finally, we obtain a small-cost solution to the \intervalKC instance, using Theorem~\ref{thm:interval-main}; this leads a solution to the \lotsizing instance, with small holding and ordering costs.

The LP obtained from \ref{MIP:lotsizing} by replacing Constraint~\eqref{LPC:integral} with $y_s \in [0, 1], \forall s \in [T]$ has unbounded integrality gap, which can be derived from the gap instance for $\KC$.   To overcome this integrality gap, we introduce a set of new inequalities. The inequalities are described in Constraint~\eqref{LPC:ls-covering}, where for convenience, we define $x_{s, {i}} = 0$ if $s > r_{i}$ and let $x_{S, {i}} = \sum_{s \in S}x_{s, {i}}$ for every $S \subseteq [T]$ and ${i} \in {[N]}$. 

\begin{figure*}[h]\begin{mdframed}
		\begin{equation}
		\min \quad \sum_{s \in [T]} y_s K_s +  \sum_{{i} \in {[N]}}d_{i} \sum_{s  \in [r_{i}]} x_{s, {i}} h_{i}(s) \qquad \text{s.t.}  \tag{$\mathsf{LP}_\mathsf{new}$}  \label{LP:lotsizing-new}
		\end{equation} \vspace*{-1.2\abovedisplayskip}\vspace*{-1.2\belowdisplayskip}
		
		\hspace*{-16pt} \begin{tabular}{p{0.3\textwidth}p{0.37\textwidth}p{0.31\textwidth}}
			{\begin{alignat}{2}
				\sum_{s \in [r_{i}]}x_{s, {i}} = 1, \forall {i} \in {[N]}; \label{LPC:ls-new-items-covered}
			\end{alignat}} &
			{\begin{alignat}{2} 
				x_{s, {i}} \in [0, 1], \forall {i} \in {[N]}, s \in [r_{i}];   \label{LPC:ls-new-x-between-01}
			\end{alignat}} &
			{\begin{alignat}{2} 
				y_s \in [0, 1], \forall s \in [T];   \label{LPC:ls-new-y-between-01}
			\end{alignat}} 
		\end{tabular} \vspace*{-25pt}
		\begin{align}
			C(S_1)  + \sum_{s \in S_2} \min\set{C_s, d(I)-C(S_1)}y_s  + \sum_{{i} \in I} x_{[T] \setminus (S_1 \cup S_2), i}d_i \geq d(I), \hspace*{0.2\textwidth}  \nonumber \\[-5pt]
			\forall S_1, S_2 \subseteq [T],  I \subseteq {[N]}  \text{ s.t. } S_1 \cap S_2 = \emptyset, C(S_1) < d(I). \label{LPC:ls-covering}
		\end{align}
	\end{mdframed}	
\end{figure*}	

We now show the validity of Constraint~\eqref{LPC:ls-covering}; focus on some sets $S_1, S_2 \subseteq [T]$ and $I \subseteq {[N]}$ such that $S_1 \cap S_2 = \emptyset$ and $C(S_1) < d(I)$.  We break $[T]$ into three sets: $S_1, S_2$ and $[T] \setminus (S_1 \cup S_2)$, and consider the quantities for the units of demands in $I$ that are satisfied by each of the 3 sets. Orders in $S_1$ satisfies at most $C(S_1)$ units of demand, orders in $S_2$ satisfies at most $\sum_{s \in S_2} C_s y_s$ units, and orders in $[T] \setminus (S_1 \cup S_2)$ satisfy exactly $\sum_{{i} \in I}x_{[T] \setminus (S_1 \cup S_2), {i}}d_{i}$ units.   Then, Constraint~\eqref{LPC:ls-covering} is valid if we replace $\min\{C_s, d(I) - C(S_1)\}$ by $C_s$.  In an integral solution we have $y_s \in \{0, 1\}$ for every $s \in S_2$.  If some $s \in S_2$ has $y_s = 1$ and $C_s > d(I) - C(S_1)$, then Constraint~\eqref{LPC:ls-covering} already holds. Thus, even if we replace $C_s$ with $\min\set{C_s, d(I) - C(S_1)}$, the constraint still holds.  In other words, we use orders in $S_2$ to cover at most $d(I)-C(S_1)$ units of demands; an order at time $s \in S_2$ has ``effective capacity'' $\min\set{C_s, d(I) - C(S_1)}$. Indeed, it is the threshold that gives the power of Constraint~\eqref{LPC:ls-covering}; without it, Constraint~\eqref{LPC:ls-covering} is implied by the constraints in the natural LP relaxation, and thus does not strengthen the LP. 

If we let $S_1 = \emptyset, S_2 = \{s\}$ and $I = \{{i}\}$ for some $s \in [T], {i} \in {[N]}$, then Constraint~\eqref{LPC:ls-covering} becomes $\min\set{C_s, d_{i}} y_s + x_{[T] \setminus \{s\}, i} d_i \geq d_i$. Since $x_{[T], i} = 1$ by Constraint~\eqref{LPC:ls-new-items-covered}, the inequality implies $\min\set{C_s, d_{i}} y_s \geq x_{s, i}d_i$. Thus, $x_{s, {i}} \leq \frac{\min\set{C_s, d_{i}}}{d_{i}}y_s \leq y_s$, implying Constraint~\eqref{LPC:ls-x-leq-y}. If we let $S=\emptyset, S_2 = \{s\}$ and $I = {[N]}$ for some $s \in [T]$, then Constraint~\eqref{LPC:ls-covering} becomes $\min\set{C_s, d([N])} y_s + \sum_{{i} \in [N]}x_{[T] \setminus \{s\}, i} d_i \geq d([N])$. Since $x_{[T], i} = 1$ for every $i$ by Constraint~\eqref{LPC:ls-new-items-covered}, the inequality implies $\min\set{C_s, d([N])} y_s \geq \sum_{i \in [N]}x_{s, i}d_i$,  which implies Constraint~\eqref{LPC:ls-capacity}. Thus, in our strengthened LP, we do not need Constraints~\eqref{LPC:ls-x-leq-y} and \eqref{LPC:ls-capacity}.  Our final strengthened LP is \ref{LP:lotsizing-new}.

\paragraph{Rounding a Fractional Solution to \ref{LP:lotsizing-new}}  Let $(x, y)$ be a fixed fractional solution satisfying Constraints~\eqref{LPC:ls-new-items-covered} to \eqref{LPC:ls-new-y-between-01}. We give a rounding algorithm that returns either an inequality of form~\eqref{LPC:ls-covering} that is violated by $y$, or a feasible solution $(x^*, y^*)$ to the \lotsizing instance whose cost is at most $10\left(\sum_{s \in [T]} y_s K_s +  \sum_{{i} \in {[N]}}d_{i} \sum_{s  \in [r_{i}]} x_{s, {i}} h_{i}(s)\right)$.

We define $\hcost(x') := \sum_{{i} \in {[N]}}d_{i} \sum_{s \in [r_{i}]}x'_{s,{i}}h_{i}(s)$, for every $x' \in [0, 1]^{[T] \times {[N]}}$ to be the holding cost of any fractional solution $(x', y')$.   We now define the requirement vector $R$ for our \intervalKC instance. If all the requirements are satisfied, then the demands can be satisfied at a small holding cost. Indeed, the approximation ratio for the holding cost is only ${\frac52}$, better than the ratio 10 for the ordering cost.  For every interval $(a, b]$ over $[T]$, the requirement $R_{a, b}$ is defined as:
\begin{align}
	R_{a, b} := \sum_{{i} \in {[N]}: r_{i} \in (a ,b]} \max\set{1-{\frac52}x_{[a], {i}}, 0}d_{i}. 
\end{align}

\begin{lemma}
	\label{lemma:sufficient-condition-for-xstar}
	Let $y^* \in \{0, 1\}^T$ and $S^*  = \{s \in [T]: y^*_s = 1\}$. If for every interval $(a, b]$ over $[T]$, we have $C(S^*\cap (a, b]) \geq R_{a, b}$, then there is an $x^*$ such that $(x^*, y^*)$ is a feasible solution to the \lotsizing instance and  $\hcost(x^*) \leq {\frac52}\hcost(x)$.
\end{lemma}
The lemma says that if all the requirements are satisfied, then we have an $x^*$ with small $\hcost(x^*)$.  In the proof, we reduce the problem of finding a good $x^*$ to the problem of finding a perfect matching in a fractional $b$-matching instance. We show that the perfect matching exists if all the requirements are satisfied.  We defer the proof of the lemma to Appendix~\ref{sec:proofs}.

Now, our goal becomes to find a set $S^* \subseteq [T]$ satisfying all the requirements. This is exactly an \intervalKC instance, and we shall apply Theorem~\ref{thm:interval-main} to solve the instance.  To guarantee the conditions of Theorem~\ref{thm:interval-main},  it suffices to guarantee that a small number of inequalities in Constraint~\eqref{LPC:ls-covering} are satisfied for $(x, y)$.
\begin{lemma}
	\label{lemma:ls-reduce-to-iKC}
	Let $(a, b]$ be an interval with $R_{a, b} > 0$, and $I = \set{{i} \in {[N]}:r_{i} \in (a, b], x_{[a], {i}} < {\frac25}}$.
	Let $(S_1, S_2)$ be an arbitrary partition of $(a, b]$ into two parts such that $C(S_1) < R_{a, b}$.  If Constraint~\eqref{LPC:ls-covering} is satisfied for $S_1, S_2$ and $I$, then we have 
	\begin{flalign}
		\text{either} && \sum_{s \in S_2}\min\set{C_s, R_{a, b} - C(S_1)} y_s &\geq R_{a, b}- C(S_1), && \label{inequ:ls-reduction-case-1} \\ 
		\text{or} && \sum_{s \in S_2: C_s \geq  R_{a, b} - C(S_1)} y_s &\geq {\frac35}. && \label{inequ:ls-reduction-case-2}
	\end{flalign}
\end{lemma}

\begin{proof}
	Notice that $R_{a, b} = \sum_{{i} \in {[N]}: r_{i} \in (a, b]}\max\set{1 - {\frac52}x_{[a], {i}}, 0}d_{i} = \sum_{{i} \in I}\big(1-{\frac52}x_{[a], {i}}\big)d_i$.  Since $x_{[r_{i}], {i}} = 1$ and $x_{s, {i}} = 0$ for every ${i} \in I, s > r_{i}$, we have $x_{(a, b], {i}} = 1- x_{[a], {i}}$ for every ${i} \in I$. Constraint \eqref{LPC:ls-covering} for $S_1, S_2$ and $I$ implies 
	\begin{align}
		\sum_{s \in S_2}\min\set{C_s, d(I) - C(S_1)} y_s &\geq \sum_{{i} \in I}x_{S_1 \cup S_2, {i}}d_{i} - C(S_1) = \sum_{{i} \in I}\left(1-x_{[a],{i}}\right)d_{i}-C(S_1).  \label{inequ:sum-large}
	\end{align}
	
	If Inequality~\eqref{inequ:ls-reduction-case-2} holds then we are done. Thus, we assume Inequality~\eqref{inequ:ls-reduction-case-2} does not hold. We consider the decrease of the sum on the left-side of Inequality~\eqref{inequ:sum-large} after we change $d(I)$ to $R_{a,b} = \sum_{{i} \in I}\big(1-{\frac52} x_{[a],{i}}\big)d_{i}$. For each $s \in S_2$, if $C_s \geq R_{a,b} - C(S_1)$,  then the decrease of the coefficient of $y_s$ is at most $d(I) - R_{a,b} = {\frac52}\sum_{{i} \in I} x_{[a], {i}}d_{i}$.  Otherwise $C_s < R_{a,b} - C(S_1)$ and there is no decrease for the coefficient of $y_s$. Since Inequality~\eqref{inequ:ls-reduction-case-2} does not hold, the decrease of the left side is at most $\sum_{s \in S_2:C_s \geq R_{a,b} - C(S_1)}y_s \times {\frac52}\sum_{{i} \in I}x_{[a],{i}}d_{i}  \leq {\frac35} \times {\frac52}\sum_{{i} \in I}x_{[a],{i}}d_{i} = {\frac32}\sum_{{i} \in I}x_{[a],{i}}d_{i}$. So, we have 
	\begin{align*}
		\sum_{s \in S_2}\min\set{C_s, R_{a,b} - C(S_1)} y_s &\geq \sum_{{i} \in I}\left(1-x_{[a],{i}}\right)d_{i}-C(S_1) - {\frac32}\sum_{{i} \in I}x_{[a],{i}}d_{i} \\
		&=\sum_{{i} \in I}\left(1-{\frac52} x_{[a],{i}}\right)d_{i} - C(S_1) = R_{a, b} - C(S_1),
	\end{align*}
	which is exactly Inequality~\eqref{inequ:ls-reduction-case-1}.
\end{proof}

Let $\hat y_s = \min\set{10 y_s, 1}$ for every $s \in [T]$ and $S^+ = \set{s \in [T], \hat y_s = 1} =  \set{s \in [T]: y_s \geq 1/10}$ be the set of elements with $y$ values at least $1/10$.  For every interval $(a, b]$, define $\tilde R_{a, b} = \max\set{R_{a, b} - C((a, b] \cap S^+),0}$ to be the residual requirement for the interval $(a, b]$, as in Theorem~\ref{thm:interval-main}. For every interval $(a, b]$ with $\tilde R_{a, b} > 0$, we define $I = \set{{i} \in {[N]}:r_{i} \in (a, b],x_{[a],{i}}< {\frac25}}$ as in Lemma~\ref{lemma:ls-reduce-to-iKC}. Let $S_1 = (a, b] \cap S^+$ and $S_2 = (a, b] \setminus S^+$.  Since $C(S_1) < R_{a, b}$, we can check if Constraint~\eqref{LPC:ls-covering} is satisfied for $I, S_1$ and $S_2$. If not, we return this violated constraint.  So, we assume the condition is satisfied. Then by Lemma~\ref{lemma:ls-reduce-to-iKC}, we have either Inequality~\eqref{inequ:ls-reduction-case-1} or \eqref{inequ:ls-reduction-case-2}.  Notice that $y_s < 1/10$ and $\hat y_s = 10y_s$ for every $s \in S_2$. Multiplying the two inequalities by 10, and replacing $R_{a, b} - C(S_1)$ with $\tilde R_{a, b}$ and $S_2$ with $(a, b] \setminus S^+$, we have
\begin{flalign*}
	\text{either}  && \sum_{s \in (a, b] \setminus S^+}\min\set{C_s, \tilde R_{a, b}} \hat y_s &\geq 10\tilde R_{a, b}, && \\
	\text{or} && \sum_{s \in (a, b] \setminus S^+: C_s \geq  \tilde R_{a, b}} \hat y_s &\geq 6.  &&
\end{flalign*} 
The above property holds for every time interval $(a, b]$ such that $\tilde R_{a, b} > 0$.  Since there are only $O(T^2)$ intervals and for each interval $(a, b]$ we only need to check one constraint of form \eqref{LPC:ls-covering}, our algorithm runs in polynomial time. 

Thus, the \intervalKC instance $(T, C,  K, R)$ and the vector $\hat y \in [0, 1]^{T}$ satisfy the condition of Theorem~\ref{thm:interval-main}.  We can apply the theorem to find a set $S^* \supseteq S^+$ such that $K(S^*) \leq \sum_{s \in [T]}\hat y_s K_s  \leq 10\sum_{s \in [T]}y_sK_s$ and $C\big([a, b) \cap S^*\big) \geq R_{a, b}$ for every interval $(a, b]$.  Let $y^* \in \{0, 1\}^T$ be the indicator vector for $S^*$; by Lemma~\ref{lemma:sufficient-condition-for-xstar} there is an $x^*$ such that $(x^*, y^*)$ is a feasible solution to the \lotsizing instance and $\hcost(x^*) \leq {\frac52}\hcost(x)$.  Thus, we obtain a feasible solution $(x^*, y^*)$ to the \lotsizing instance whose cost is at most $10 \sum_{s \in [T]}y_s K_s + {\frac52}\hcost(x)$; this finishes the proof of Theorem~\ref{thm:lot-sizing}.
 	\section{Approximation Algorithm for Laminar Knapsack Covering via Iterative Rounding}
\label{sec:laminarKC}

This section is dedicated to the proof of Theorem~\ref{thm:laminar-main}, by describing our iterative rounding algorithm. 	Recall that we are given a \laminarKC instance $(T, C, K, \calS, R)$, a vector $y \in [0, 1]^T$ and set $S^+ = \set{s \in {[T]}: y_s = 1}$. Let $\tilde R_{a, b} := \max\set{R_{a, b}- C((a, b] \cap S^+), 0}$ for every $(a, b] \in \calS$. For every $(a, b] \in \calS$ with $\tilde R_{a, b}> 0$, either Inequality~\eqref{inequ:laminar-main-requirement-1} or Inequality~\eqref{inequ:laminar-main-requirement-2} holds. 

We now give an overview of the rounding algorithm.  We maintain a set $S^* \subseteq {[T]}$ of knapsacks that we already selected, and a set $S^\circ \subseteq {[T]}$ of knapsacks we discarded. There is a sub-family $\calS^1 \cup \calS^2 \subseteq \calS$ of ``active'' intervals from the laminar family $\calS$, where $\calS^1 \cap \calS^2 = \emptyset$. Each set $(a, b] \in \calS^1 \cup \calS^2$ has positive residual requirement $\hat R_{a, b}$, after we selected knapsacks in $S^*$; that is, $\hat R_{a, b}:=R_{a, b} - C\big(S^* \cap (a, b]\big) > 0$. We maintain an LP relaxation (\ref{LP:laminar-iteration}) in the rounding algorithm, in which we have a constraint for every interval $(a, b] \in \calS^1 \cup \calS^2$. We maintain a feasible solution $y$ to the LP. In each iteration of the rounding algorithm, we update $y$ to be an optimum vertex-point solution of \ref{LP:laminar-iteration}.  Then, we show that there must be some knapsack $s \notin S^\circ \cup S^*$ such that $y_s \in \{0, 1\}$. Depending on whether $y_s = 0$ or $y_s = 1$, we discard or select $s$, by adding $s$ to $S^\circ$ or $S^*$. Then we can update $\calS^1, \calS^2$ and the residual requirement vector $\hat R$ accordingly.  To analyze the algorithm, we prove two key lemmas. First, we show that the feasibility of $y$ is maintained, after we update $S^\circ, S^*, \calS^1, \calS^2$ and $\hat R$. Second,  we show that the algorithm makes progress in each iteration: some new knapsack $s$ is added to $S^\circ$ or $S^*$.

Now we describe the LP relaxation \ref{LP:laminar-iteration} used in the iterative rounding. In the LP,  $y_s$ indicates whether $s$ is included in the final solution.   Thus, we require $y_s = 0$ for every $s \in S^\circ$ (Constraint~\eqref{LPC:laminar-iteration-Scirc}) and $y_s = 1$ for every $s \in S^*$ (Constraint~\eqref{LPC:laminar-iteration-Sstar}).   For a set $(a, b] \in \calS^1$, we require Constraint~\eqref{LPC:laminar-iteration-calS1} to hold; for a set $(a, b] \in \calS^2$, we require Constraint~\eqref{LPC:laminar-iteration-calS2} to hold. Notice that the two constraints are respectively Inequality~\eqref{inequ:laminar-main-requirement-1} and  Inequality~\eqref{inequ:laminar-main-requirement-2}, with $S^+$ replaced by $S^*$ and $\tilde R$ replaced by $\hat R$. \smallskip

\begin{figure*}[h]\begin{mdframed}
	\begin{flalign}
		&& \min \qquad \sum_{s \in {[T]}} y_sK_s  \qquad \text{ s.t. }  \tag{$\mathsf{LP}_{\mathsf{IterRound}}$} \label{LP:laminar-iteration} &&
	\end{flalign} \vspace*{-\abovedisplayskip}\vspace*{-\belowdisplayskip}
	
	\noindent\begin{tabular}{p{0.62\textwidth}p{0.35\textwidth}}
		{\begin{alignat}{2} 
			\sum_{s \in (a, b] \setminus S^*}\min\set{C_s, \hat R_{a, b}}y_s &\geq 2 \hat R_{a, b},  &\qquad \forall (a, b] &\in \calS^1; \label{LPC:laminar-iteration-calS1} \\
			\sum_{s \in (a, b] \setminus S^*: C_s \geq \hat R_{a, b}}y_s &\geq 1, &\qquad \forall (a, b] &\in \calS^2; \label{LPC:laminar-iteration-calS2} 
		\end{alignat}}
		&
		{\begin{alignat}{2} 
			y_s &= 0, \qquad &\forall s &\in S^\circ; \label{LPC:laminar-iteration-Scirc} \\
			y_s &= 1, \qquad &\forall s &\in S^*; \label{LPC:laminar-iteration-Sstar} \\
			y_s &\in [0, 1], \qquad &\forall s &\notin S^\circ \cup S^*.  \label{LPC:laminar-iteration-Sother}
		\end{alignat}}
	\end{tabular}\vspace*{-\belowdisplayskip}
\end{mdframed}\end{figure*}

The pseudo-code for the rounding algorithm is given in Algorithm~\ref{alg:alg-laminar-KC}. We initialize $S^*, S^\circ,  \hat R,  \calS^1$ and $\calS^2$ in Statements~\ref{STATE:laminar-init-S} to \ref{STATE:laminar-init-calS1}. In each iteration of the outer loop (the loop beginning with Statement~\ref{STATE:laminar-loop}), we first repeatedly remove intervals $(a, b]$ from $\calS^1 \cup \calS^2$ if the requirement for $(a, b]$ is implied by the requirement for some other interval $(a', b'] \in \calS^1 \cup \calS^2$ (Statements~\ref{STATE:laminar-check-remove} and \ref{STATE:laminar-remove-2}). Then, we solve \ref{LP:laminar-iteration} to obtain an optimum vertex point solution $y$ (Statement~\ref{STATE:laminar-solve-lp}). Some knapsacks $s \in [T] \setminus S^\circ$ may have $y_s = 0$, in which case we permanently discard $s$ by adding $s$ to $S^\circ$ (Statement~\ref{STATE:laminar-update-Scirc}); some knapsacks $s \in [T] \setminus S^*$ may have $y_s = 1$, in which case we add $s$ to our final solution $S^*$ (Statement~\ref{STATE:laminar-update-Sstar}).  With $S^*$ updated, we shall update $\hat R$, $\calS^1$ and $\calS^2$ accordingly (Statements~\ref{STATE:laminar-loop-for-(a, b]} to \ref{STATE:laminar-move}).  In particular, we remove intervals $(a, b]$ from $\calS^1 \cup \calS^2$ if the requirement for $(a, b]$ is already satisfied by $S^*$ (Statement~\ref{STATE:laminar-remove-1}).  We move sets $(a, b]$ from $\calS^1$ to $\calS^2$ if $(a, b]$ satisfies Constraint~\eqref{LPC:laminar-iteration-calS2} (Statement~\ref{STATE:laminar-move}).  The algorithm terminates when $\calS^1 \cup \calS^2  = \emptyset$.  

Since we removed intervals from $\calS^1 \cup \calS^2$ in Statements~\ref{STATE:laminar-check-remove}, \ref{STATE:laminar-remove-2} and Statement~\ref{STATE:laminar-remove-1}, the requirements for intervals $(a, b]$ in $\calS \setminus (\calS^1 \cup \calS^2)$ are irrelevant. For such an interval $(a, b]$, either $C(S^* \cap (a, b]) \geq R_{a, b}$, or there exists some set $(a', b'] \in \calS^1 \cup \calS^2$ such that $(a', b'] \setminus (S^\circ \cup S^*) = (a, b] \setminus (S^\circ \cup S^*)$ and $\hat R_{a', b'} = R_{a', b'} - C(S^* \cap (a', b']) \geq R_{a, b} - C(S^* \cap (a, b])$. In the former case, the requirement for $(a, b]$ is already satisfied; in the later case, the requirement for $(a, b]$ is implied by the requirement for some interval $(a', b'] \in \calS^1 \cup \calS^2$, conditioned on that we must choose knapsacks in $S^*$ and not choose knapsacks in $S^\circ$. So, when $\calS^1 \cup \calS^2  = \emptyset$ becomes empty, the requirements for all sets in $\calS$ are satisfied.

\begin{algorithm}
	\caption{Iterative Rounding Algorithm for the Proof of Theorem~\ref{thm:laminar-main}.}
	\label{alg:alg-laminar-KC}
	\textbf{Input}: $(T, C, K, \calS, R), y, S^+$ and $\left(\tilde R_{a, b}\right)_{(a, b] \in \calS}$ as in Theorem~\ref{thm:laminar-main}. \\
	\textbf{Output}: a feasible solution $S^* \subseteq {[T]}$ to the instance with $K\big(S^*\big) \leq \sum_{s \in {[T]}}y_s K_s$. 
	
	\begin{algorithmic}[1]
				\STATE let $S^\circ \gets \emptyset, \quad S^* \gets S^+, \quad \hat R_{a, b} \gets \tilde R_{a, b}$ for every $(a, b] \in \calS$  \label{STATE:laminar-init-S}
		\STATE let $\calS^2 \gets \set{(a, b] \in \calS: \hat R_{a, b} > 0, \text{Constraint~\eqref{LPC:laminar-iteration-calS2} is satisfied for } (a, b]}$\label{STATE:laminar-init-calS2}
		\STATE let $\calS^1 \gets \set{(a, b] \in \calS: \hat R_{a, b} > 0} \setminus \calS^2$ \label{STATE:laminar-init-calS1}
		\STATE \textbf{while} $\calS^1 \cup \calS^2 \neq \emptyset$ \textbf{do} \hfill (\textbf{outer loop}) \label{STATE:laminar-loop} 
		\STATE\hspace*{\algorithmicindent} \textbf{while} there exist two distinct intervals $(a, b], (a', b'] \in \calS^1 \cup \calS^2$ \\ \hspace*{0.1\textwidth} s.t. $(a, b] \setminus (S^\circ \cup S^*) = (a', b'] \setminus (S^\circ \cup S^*)$  and $\hat R_{a', b'} \geq \hat R_{a, b}$ \textbf{do} \label{STATE:laminar-check-remove}
		\STATE \hspace*{2\algorithmicindent} $\calS^1 \gets \calS^1 \setminus \{(a, b]\}, \calS^2 \gets \calS^2 \setminus \{(a, b]\}$	\label{STATE:laminar-remove-2}
		\STATE\hspace*{\algorithmicindent} let $y$ to be an optimal vertex point solution of \ref{LP:laminar-iteration} \label{STATE:laminar-solve-lp}			
		\STATE\hspace*{\algorithmicindent} \textbf{while} there exists $s \in {[T]} \setminus S^\circ$ such that $y_s = 0$ \textbf{do}: add $s$ to $S^\circ$ \label{STATE:laminar-update-Scirc}
		\STATE\hspace*{\algorithmicindent} \textbf{while} there exists $s \in {[T]} \setminus S^*$ such that $y_s = 1$ \textbf{do} \hfill (\textbf{middle loop}) \label{STATE:laminar-check-ys1} 
		\STATE\hspace*{2\algorithmicindent} add $s$ to $S^*$ \label{STATE:laminar-update-Sstar} 		\STATE\hspace*{2\algorithmicindent} \textbf{for} every $(a, b] \in \calS^1 \cup \calS^2$ such that $s \in (a, b]$ \textbf{do} \label{STATE:laminar-loop-for-(a, b]} \hfill (\textbf{inner loop})
		\STATE\hspace*{3\algorithmicindent} $\hat R_{a, b} \gets \hat R_{a, b} - C_s$ \label{STATE:laminar-update-R}
				\STATE\hspace*{3\algorithmicindent} \textbf{if} $\hat R_{a, b}  \leq  0$ \textbf{then} $\calS^1 \gets \calS^1 \setminus \{(a, b]\}, \calS^2 \gets \calS^2 \setminus \{(a, b]\}$ \label{STATE:laminar-remove-1}
		\STATE\hspace*{3\algorithmicindent} \textbf{if} $(a, b] \in \calS^1$ and Constraint~\eqref{LPC:laminar-iteration-calS2} holds for $(a, b]$ \textbf{then} $\calS^1 \gets \calS^1 \setminus \{(a, b]\}, \calS^2 \gets \calS^2 \cup \{(a, b]\}$ \label{STATE:laminar-move}
	\end{algorithmic}
\end{algorithm}

Before formally proving the two key lemmas we need to prove Theorem~\ref{thm:laminar-main}, we make some simple observations about Algorithm~\ref{alg:alg-laminar-KC}, assuming Statement~\ref{STATE:laminar-solve-lp} always finds a feasible solution $y$.  \begin{observation} \label{obs:laminar-alg}
	\begin{properties}{obs:laminar-alg}
		\item After the initialization of $S^\circ$ and $S^*$ in Statement~\ref{STATE:laminar-init-S}, we always have $y_s = 0$ for every $s \in S^\circ$ and $y_s = 1$ for every $s \in S^*$. \label{obs:Sstar-Scirc}
		\item After the initialization of $\calS^1$ and $\calS^2$ in Statements~\ref{STATE:laminar-init-calS2} and \ref{STATE:laminar-init-calS1}, we always have $\calS^1  \cap \calS^2 = \emptyset$. \label{obs:calS1-calS2-disjoint}
		\item At the beginning of each iteration of the outer loop, we have $\hat R_{a, b} = R_{a, b} - C(S^* \cap (a, b]) > 0$ for every $(a, b] \in \calS^1 \cup \calS^2$. \label{obs:hatR}
	\end{properties}
\end{observation}

Observation~\ref{obs:Sstar-Scirc} holds since we add $s$ to $S^\circ$ only if $y_s = 0$, to $S^*$ only if $y_s = 1$. In \ref{LP:laminar-iteration}, we have
Constraints \eqref{LPC:laminar-iteration-Scirc} and \eqref{LPC:laminar-iteration-Sstar}. Thus the two constraints hold when we update $y$ in Statement~\ref{STATE:laminar-solve-lp}.  After Statement~\ref{STATE:laminar-init-calS1}, $\calS^1 \cap \calS^2 = \emptyset$. The only place we add an interval to either $\calS^1$ or $\calS^2$ is Statement~\ref{STATE:laminar-move}, in which we move the interval from $\calS^1$ to $\calS^2$. Thus Observation~\ref{obs:calS1-calS2-disjoint} holds. After Statement~\ref{STATE:laminar-init-calS1}, we have that for every interval $(a, b] \in \calS^1 \cup \calS^2$, $\hat R_{a, b} = \tilde R_{a, b} = \max\set{R_{a, b} - C(S^+ \cap (a, b]), 0} = \max\set{R_{a, b} - C(S^* \cap (a, b]), 0} > 0$.  Thus, $\hat R_{a, b} = R_{a, b} - C(S^* \cap (a, b]) > 0$. Every time we add a knapsack $s$ to $S^*$ in Statement~\ref{STATE:laminar-update-Sstar},  for every $(a, b] \in \calS^1 \cup \calS^2$ such that $s \in (a, b]$, we decrease $\hat R_{a, b}$ by $C_s$ in Statement~\ref{STATE:laminar-update-R}. Then if $\hat R_{a, b} \leq 0$ after the decrease, we remove $(a, b]$ from $\calS^1 \cup \calS^2$ in Statement~\ref{STATE:laminar-remove-1}. Thus, Observation~\ref{obs:hatR} holds.

The first key lemma is Lemma~\ref{lemma:laminar-maintain-feasible}. The heart of the proof of the lemma is in the proof of Lemma~\ref{lemma:middle-loop-maintain-feasibility}. We prove Lemma~\ref{lemma:middle-loop-maintain-feasibility} now, while deferring the proof of Lemma~\ref{lemma:laminar-maintain-feasible} to Appendix~\ref{sec:proofs}.
\begin{lemma}
	\label{lemma:laminar-maintain-feasible}
	At the beginning of each iteration of the outer loop, $y$ is a feasible solution to \ref{LP:laminar-iteration}.
\end{lemma}

\begin{lemma}
	\label{lemma:middle-loop-maintain-feasibility}
	If $y$ is a feasible solution to \ref{LP:laminar-iteration} at the beginning of an iteration of the middle loop (the loop beginning with Statement~\ref{STATE:laminar-check-ys1}), then it is also feasible at the end of the iteration. 
\end{lemma}

\begin{proof}
	For notational purposes, we use $s^*$ to denote the knapsack that will be added to $S^*$ in this iteration.  After this iteration, Constraint~\eqref{LPC:laminar-iteration-Scirc} is unaffected, and Constraints~\eqref{LPC:laminar-iteration-Sstar} and \eqref{LPC:laminar-iteration-Sother} remain satisfied since $y_{s^*} = 1$.  Thus, we only need to focus on Constraints~\eqref{LPC:laminar-iteration-calS1} and \eqref{LPC:laminar-iteration-calS2}.  If some $(a, b] \in \calS^1 \cup \calS^2$ does not contain $s^*$, then $(a, b] \setminus S^*$ and $\hat R_{a, b}$ do not change in the iteration. Thus the constraint for $(a, b]$ (either Constraint~\eqref{LPC:laminar-iteration-calS1} or Constraint~\eqref{LPC:laminar-iteration-calS2}) is unaffected. Thus, we can focus on an interval $(a, b] \in \calS^1 \cup \calS^2$ that contains $s^*$. In the iteration of the inner loop (the loop beginning with Statement~\ref{STATE:laminar-loop-for-(a, b]}) for this interval $(a, b]$, Statement~\ref{STATE:laminar-update-R} decreases $\hat R_{a, b}$ by $C_{s^*}$, Statement~\ref{STATE:laminar-remove-1} removes $(a, b]$ from $\calS^1 \cup \calS^2$ if $\hat R_{a, b}$ becomes at most $0$, and Statement~\ref{STATE:laminar-move} moves $(a, b]$ from $\calS^1$ to $\calS^2$ if $(a, b]$ satisfies Constraint~\eqref{LPC:laminar-iteration-calS2}.
	
	To the end of this proof, $S^*$ will refer to the set $S^*$ before we add $s^*$, $\hat R_{a, b}$ will refer to the value of $\hat R_{a, b}$ before we run Statement~\ref{STATE:laminar-update-R} and $\hat R^{\mathsf{new}}_{a, b} = \hat R_{a, b} - C_{s^*}$ will be the value of $\hat R_{a, b}$ after Statement~\ref{STATE:laminar-update-R}. If $\hat R^{\mathsf{new}}_{a, b} \leq 0$, then $(a, b]$ will be removed from $\calS^1 \cup \calS^2$ and there will be no Constraint for $(a, b]$ in \ref{LP:laminar-iteration}. Thus, we can assume $\hat R^{\mathsf{new}}_{a, b} > 0$. 
	
	Consider the first case where we have $(a, b] \in \calS^1$ at the beginning of the iteration of the middle loop. If $\sum_{s \in (a, b] \setminus S^* \setminus \{s^*\}: C_s \geq \hat R^{\mathsf{new}}_{a, b}}y_s \geq 1$, then $(a, b]$ will be moved to $\calS^2$ and Constraint~\eqref{LPC:laminar-iteration-calS2} for $(a, b]$ will be satisfied. So,  we can assume that $\sum_{s \in (a, b] \setminus S^* \setminus \{s^*\}: C_s \geq \hat R^{\mathsf{new}}_{a, b}} y_s < 1$.  $(a, b]$ will remain in $\calS^1$ after the iteration. We consider how Constraint~\eqref{LPC:laminar-iteration-calS1} for $(a, b]$ is affected by adding $s^*$ to $S^*$ and changing $\hat R_{a, b}$ to $\hat R^{\mathsf{new}}_{a, b}$. The right side of the inequality is decreased by exactly $2\big(\hat R_{a, b} - \hat R^{\mathsf{new}}_{a, b}\big) = 2C_{s^*}$.  We now consider the decrease of the left side.  First, adding $s^*$ to $S^*$ will decrease the left side by $y_{s^*} \min \set{C_{s^*}, \hat R_{a, b}} = C_{s^*}$ since $y_{s^*} = 1$ and $C_{s^*} < \hat R_{a, b}$.  Second, some knapsack $s \in (a, b] \setminus S^* \setminus \{s^*\}$ will have $\min\set{C_s, \hat R^{\mathsf{new}}_{a, b}} < \min\set{C_s, \hat R_{a, b}}$. This happens only if $C_s \geq \hat R^{\mathsf{new}}_{a,b}$. Moreover, if this happens, the decrease of the left-side due to this $s$ is at most $y_s(\hat R_{a, b} - \hat R^{\mathsf{new}}_{a, b}) = y_s C_{s^*}$.  Since we have $\sum_{s \in (a, b] \setminus S^* \setminus \{s^*\}: C_s \geq \hat R^{\mathsf{new}}_{a,b}} y_s< 1$, the overall decrease of the left-side of Constraint~\eqref{LPC:laminar-iteration-calS1} for $(a, b]$ is at most $C_{s^*} + \sum_{s \in (a, b] \setminus S^*\setminus \{s^*\}: C_s \geq \hat R^{\mathsf{new}}_{a,b}}y_sC_{s^*}  < 2C_{s^*}$, which is the decrease of its right-side. Thus, the constraint for $(a, b]$ remains satisfied at the end of the iteration for the middle loop. 
	
	Then assume that $(a, b] \in \calS^2$ at the beginning of the iteration of the inner loop.  Notice that we have assumed that $\hat R^{\mathsf{new}}_{a, b} > 0$, which implies $C_{s^*} < \hat R_{a, b}$. The left-side of Constraint~\eqref{LPC:laminar-iteration-calS2} can only increase:  (i) though we will add the knapsack $s^*$ to $S^*$, we have $C_{s^*} < \hat R_{a, b}$  and thus it does not contribute to the left-side at the beginning of the iteration of the middle loop; (ii) for a knapsack $s \in (a, b] \setminus S^*\setminus \{s^*\}$, $C_s \geq \hat R_{a, b}$ implies that $C_s \geq \hat R^{\mathsf{new}}_{a,b}$.   Thus, Constraint~\eqref{LPC:laminar-iteration-calS2} for $(a, b]$ remains true at the end of the iteration of the middle loop.  This finishes the proof of the lemma. 
\end{proof}				

We defer  the proof of the second key lemma to Appendix~\ref{sec:proofs}. The proof uses the fact that $\calS$ is a laminar family and the properties of vertex-point solutions to \ref{LP:laminar-iteration}.
\begin{lemma}
	\label{lemma:terminate-in-T-iterations}
	The outer loop will terminate in at most $T$ iterations. 
\end{lemma}

With the two key lemmas, we can complete the proof of Theorem~\ref{thm:laminar-main}.  By Observation~\ref{obs:Sstar-Scirc}, $y_s = 1$ for every $s \in S^*$. So, we always have $K(S^*) \leq \sum_{s \in {[T]}}y_s K_s$. The only statement that changes $y$ is Statement~\ref{STATE:laminar-solve-lp}. Since $y$ is a feasible solution to \ref{LP:laminar-iteration} before running the statement, and the LP tries to minimize $\sum_{s \in {[T]}}y_sK_s$, we have that $\sum_{s \in {[T]}}y_sK_s$ can only decrease over the course of the algorithm.  Thus the returned solution $S^*$ has cost at most $\sum_{s \in {[T]}}y_s K_s$, for the initial $y$-vector.  

It remains to show that $S^*$ is a feasible solution to the instance $(T, C, K, \calS, R)$.  Notice that we remove intervals $(a, b]$ from $\calS^1 \cup \calS^2$ in Statement~\ref{STATE:laminar-remove-2} and Statement~\ref{STATE:laminar-remove-1}.   Assume towards the contradiction that $S^*$ is not a feasible solution. Let $(a, b] \in \calS$ be an interval with $C((a, b] \cap S^*) < R_{a, b}$; if there are many such intervals $(a, b]$, we choose the one that is removed from $\calS^1 \cup \calS^2$ the latest.  If we removed $(a, b]$ from $\calS^1 \cup \calS^2$ in Statement~\ref{STATE:laminar-remove-1}, then at that time we already have $C((a, b] \cap S^*) \geq R_{a, b}$. Thus $(a, b]$ can only be removed from $\calS^1 \cup \calS^2$ in Statement~\ref{STATE:laminar-remove-2}.  By Observation~\ref{obs:hatR}, $\hat R_{a', b'} = R_{a', b'} - C((a', b'] \cap S^*) > 0$ for every $(a', b'] \in \calS^1 \cup \calS^2$, at the beginning of an iteration of the outer loop. At the time of the removal there exists some other $(a', b'] \in \calS^1 \cup \calS^2$ such that $(a, b] \setminus (S^\circ \cup S^*) = (a', b'] \setminus (S^\circ \cup S^*)$ and $R_{a', b'} - C((a', b'] \cap S^*) \geq R_{a, b} - C((a, b] \cap S^*)$.  The inequality remains true as the algorithm proceeds since whenever we add a new knapsack $s \notin S^\circ \cup S^*$ to $S^*$,  $s \in (a, b]$ if and only if $s \in (a', b']$.   By our choice of $(a, b]$, at the end of the algorithm we have $C((a', b'] \cap S^*) \geq R_{a', b'}$, implying that $C((a, b] \cap S^*) \geq R_{a, b}$, a contradiction.  Thus, $S^*$ is a feasible solution and we proved Theorem~\ref{thm:laminar-main}.
 	
	\newpage
	
	\bibliographystyle{plain}
	\bibliography{reflist}
	
	\appendix	
	\section{Reduction of Interval Knapsack Covering to Laminar Knapsack Covering}
\label{sec:intervalKC}

In this section, we prove Theorem~\ref{thm:interval-main} via a reduction of the given \intervalKC instance to a \laminarKC instance. Recall the given \intervalKC instance is $(T, C, K, R)$. We are also given a vector $y \in [0, 1]^{[T]}$ and $S^+ = \set{s \in [T]:y_s = 1}$. For every interval $(a, b]$, we have $\tilde R_{a, b} = \max\set{R_{a, b}- C((a, b] \cap S^+), 0}$; if $\tilde R_{a, b} > 0$, then either Inequality~\eqref{inequ:interval-main-requirement-1} or Inequality~\eqref{inequ:interval-main-requirement-2} holds. 

We start by defining the requirement function for the \laminarKC instance.  For every interval $(a, b]$, we define
\begin{equation}
	\tilde R'_{a, b} := \sup\set{W \geq 0: \sum_{s  \in (a, b] \setminus S^+}\min\set{C_s, W}y_s \geq 2W, \text{ or } \sum_{s \in (a, b] \setminus S^+: C_s \geq W}y_s \geq 1}. \label{equ:define-Rprime}
\end{equation}

That is, $\tilde R'_{a, b}$ is the largest possible value that makes either Inequality~\eqref{inequ:laminar-main-requirement-1} or Inequality~\eqref{inequ:laminar-main-requirement-2} hold for the interval $(a, b]$.  To construct a laminar family $\calS$ of intervals for the \laminarKC instance, we run a simple recursive procedure: let $\calS \gets \emptyset$ initially and then we call \textsf{construct-laminar-family}$(0, T)$. 

\begin{algorithm}[h]
	\caption{\textsf{construct-laminar-family}$(a, b)$} \label{alg:construct-laminar}
	\begin{algorithmic}[1]
		\STATE add $(a, b]$ to $\calS$ 		\STATE \textbf{if} $a + 1< b$ \textbf{then} 
		\STATE\hspace{\algorithmicindent} find the $c \in (a, b)$ that maximizes $\min\set{\tilde R'_{a, c}, \tilde R'_{c, b}}$ \label{STATE:interval-find-c}
		\STATE\hspace{\algorithmicindent} \textsf{construct-laminar-family}$(a, c)$
		\STATE\hspace{\algorithmicindent} \textsf{construct-laminar-family}$(c, b)$
	\end{algorithmic}
\end{algorithm}

The $\calS$ we constructed naturally defines a laminar tree, in which every non-leaf interval $(a, b]$ has exactly two children $(a, c]$ and $(c, b]$, for some $c \in (a, b)$; for every leaf $(a, b]$ in the tree, we have $b - a = 1$. The next lemma shows that the $\tilde R'$-requirements for intervals in $\calS$ will capture the $\tilde R$-requirements for all intervals. 
\begin{lemma}
	\label{lemma:calS-is-good}
	For every interval $(a, b]$ such that $\tilde R_{a, b} > 0$, there exist an interval $(a', b'] \in \calS$ such that $(a', b'] \subseteq (a, b]$ and $\tilde R'_{a', b'} \geq \tilde R_{a, b}$.
\end{lemma}

\begin{figure}[h]
	\centering
	\includegraphics[width=0.8\textwidth]{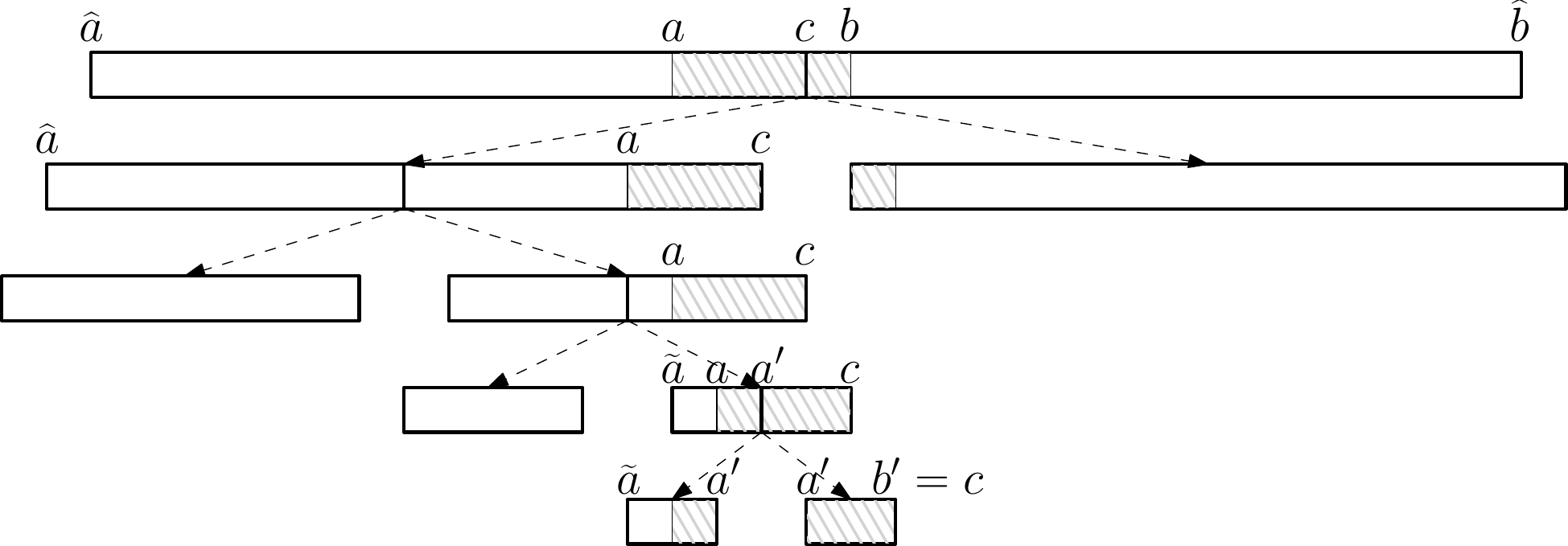}
	\caption{Illustration for the Proof of Lemma~\ref{lemma:calS-is-good}.}
	\label{fig:laminar-tree}
\end{figure}

\begin{proof}
	We first give the high-level idea behind the proof; see Figure~\ref{fig:laminar-tree} for illustration.  We consider the  inclusive-minimal interval $(\hat a, \hat b] \in \calS$ that contains $(a, b]$. Let $(\hat a, c]$ and $(c, \hat b]$ be the two child-intervals of $(\hat a, \hat b]$.  Then, $c \in (a, b)$ by our choice of $(\hat a, \hat b)$. One of the two intervals in $(a, c]$ and $(c, b]$ must contain enough total capacity.  Assume this interval is $(a, c] \subseteq (\hat a, c)$. Then, we start from $(\tilde a, c] =  (\hat a, c]$; if the right-child interval $(a', c]$ of $(\tilde a, c]$ is a superset of $(a, c]$, then we let $(\tilde a, c] = (a', c]$ and repeat. So, eventually, we can find an interval $(\tilde a, c]$ in the tree, and its right child $(a', c]$, such that $(\tilde a, c] \supseteq (a, c] \supsetneq (a', c]$. Then, we can show that the interval $(a', b'] = (a', c]$ satisfies the requirement of the lemma.
	
	Now we prove the lemma formally. We consider the first case in which we have Inequality~\eqref{inequ:interval-main-requirement-1} for $(a, b]$.  Then $10\tilde R_{a, b} \leq \sum_{s \in (a, b] \setminus S^+} \min\set{C_s, \tilde R_{a, b}}  y_s  \leq \sum_{s \in (a, b] \setminus S^+} \tilde R_{a, b} \cdot y_s$. Since $\tilde R_{a, b} > 0$ and $y_s \leq 1$ for every $s \in (a, b] \setminus S^+$, we have $b - a \geq 10$.  
	
	Consider the laminar tree defined by $\calS$. Let $\big(\hat a, \hat b\big]$ be the inclusive-minimal set in $\calS$ such that $(a, b] \subseteq \big(\hat a, \hat b\big]$. Such a set exists since $[T] \in \calS$.   Since $\hat b - \hat a \geq b - a \geq 10$, there will be two children $(\hat a, c]$ and $(c, \hat b]$ of $\big(\hat a, \hat b\big]$ in the laminar tree. By our choice of $\big(\hat a, \hat b\big]$, neither $\big({\hat a}, c\big]$ nor $\big(c, {\hat b}\big]$ is a superset of $(a, b]$. Thus, $a < c < b$.  By Inequality~\eqref{inequ:interval-main-requirement-1}, we have $\sum_{s \in (a, b] \setminus S^+}\min\set{C_s,\tilde R_{a,b}}y_s \geq 10\tilde R_{a,b}$. Thus, we have either $\sum_{s \in (a, c] \setminus S^+}\min\set{C_s, \tilde R_{a,b}}y_s  \geq 5 \tilde R_{a, b}$ or $\sum_{s = (c, b] \setminus S^+}\min\set{C_s, \tilde R_{a,b}}y_s  \geq 5 \tilde R_{a, b}$. Without loss of generality, we assume the first inequality holds; this implies that $c - a \geq 5$. 
	
	We run the following procedure to find two sets $\big(\tilde a, c\big], \big(a', c\big] \in \calS$. Let $\tilde a = \hat a$ initially; thus $\big(\tilde a, c\big] \in \calS$ and $\big(\tilde a, c\big] \supseteq \big(a, c\big]$. While the right child $\big(a', c\big]$ of $\big(\tilde a, c\big]$ is a superset of $(a, c]$, we let $\tilde a = a'$ and repeat.  At the end of the process, we find a set $(\tilde a ,c] \in \calS$ and its right child $(a', c] \in \calS$ such that $\big(\tilde a, c\big] \supseteq (a, c] \supsetneq \big(a', c\big]$. (Notice that the right child of $\big(\tilde a, c\big]$ always exists during the procedure since $c - a \geq 5$).  Notice that $\sum_{s \in (\tilde a, c] \setminus S^+}\min \set{C_s, \tilde R_{a, b}} y_s \geq 5 \tilde R_{a, b}$ and $\min\set{C_s, \tilde R_{a, b}}y_s \leq \tilde R_{a, b}$ for every $s \in (\tilde a, c] \setminus S^+$.  There is a number $\ell \in (\tilde a, c)$ such that $\sum_{s \in (\tilde a, \ell] \setminus S^+ } \min\set{C_s, \tilde R_{a, b}}y_s  \geq 2 \tilde R_{a, b}$ and $\sum_{s \in (\ell, c]\setminus S^+}\min\set{C_s, \tilde R_{a, b}}y_s \geq 2 \tilde R_{a, b}$. Thus, $\tilde R'_{\tilde a, \ell} \geq \tilde R_{a, b}$ and $\tilde R'_{\ell, c} \geq \tilde R_{a, b}$, by the definition of $\tilde R'$ in Equation~\eqref{equ:define-Rprime}. By the way we choose $c$ in Statement~\ref{STATE:interval-find-c} of Algorithm~\ref{alg:construct-laminar}, we have that $\min\set{\tilde R'_{\tilde a, a'}, \tilde R'_{a', c}} \geq \min\set{\tilde R'_{\tilde a, \ell}, \tilde R'_{\ell, c}} \geq \tilde R_{a, b}$. In particular, $\tilde R'_{a', c} \geq \tilde R_{a, b}$.  Let $b' = c$; then $ (a', b'] = (a', c] \subseteq (a, c] \subseteq (a, b]$ and $\tilde R'_{a', b'} = \tilde R'_{a', c} \geq \tilde R_{a, b}$.  This finishes the proof of the Lemma if Inequality~\eqref{inequ:interval-main-requirement-1} is satisfied for $(a, b]$. 
	
	We now consider the second case in which Inequality~\eqref{inequ:interval-main-requirement-2} is satisfied for $(a, b]$. The proof for this case is very similar to the previous case and thus we only give a sketch.  Let $\big(\hat a, \hat b \big]$ be the inclusive-minimal set in $\calS$ that contains $(a, b]$. Similarly, we can prove $\hat b - \hat a \geq b - a \geq 6$ and $(\hat a, \hat b]$ has two children $(\hat a, c]$ and $(c, \hat b]$ in the laminar tree. By the way we choose $(\hat a, \hat b]$, we have $a < c < b$.  Inequality~\eqref{inequ:interval-main-requirement-2} implies $\sum_{s \in (a, b] \setminus S^+: C_s \geq \tilde R_{a, b}} y_s \geq 6$. Thus, either $\sum_{s \in (a, c] \setminus S^+: C_s \geq \tilde R_{a, b}} y_s \geq 3$ or $\sum_{s \in (c, b] \setminus S^+: C_s \geq \tilde R_{a, b}} y_s \geq 3$. W.l.o.g, we assume the first case happens.  We find a set $(\tilde a, c] \in \calS$ and its right child $(a', c] \in \calS$ as in the previous case.  So, $\big(\tilde a, c\big] \supseteq (a, c] \supsetneq \big(a', c\big]$. Since $\sum_{s \in (\tilde a, c] \setminus S^+: C_s \geq \tilde R_{a, b}} y_s \geq 3$ and $y_s \leq 1$ for every $s \in (\tilde a, c] \setminus S^+$, there is an $\ell \in (\tilde a, c)$ such that $\sum_{s \in (\tilde a, \ell] \setminus S^+: C_s \geq \tilde R_{a, b}} y_s \geq 1$ and $\sum_{s \in (\ell, c] \setminus S^+: C_s \geq \tilde R_{a, b}} y_s \geq 1$. By the definition of $\tilde R'$ in Equation~\eqref{equ:define-Rprime}, we have $\tilde R'_{\tilde a, \ell} \geq \tilde R_{a, b}$ and $\tilde R'_{\ell, c} \geq \tilde R_{a, b}$. By the way we select $c$ in Statement~\ref{STATE:interval-find-c} of Algorithm~\ref{alg:construct-laminar}, we have that $\min\set{\tilde R'_{\tilde a, a'}, \tilde R'_{a', c}} \geq \min\set{\tilde R'_{\tilde a, \ell}, \tilde R'_{\ell, c}} \geq \tilde R_{a, b}$.  Let $b' = c$; then $ (a', b'] = (a', c] \subseteq (a, c] \subseteq (a, b]$ and $\tilde R'_{a', b'} = \tilde R'_{a', c} \geq \tilde R_{a, b}$. This finishes the proof of the lemma for the second case.
\end{proof}

With Lemma~\ref{lemma:calS-is-good}, we can define our \laminarKC instance.  Let $R'_{a, b} = \tilde R'_{a, b} + C\big((a, b] \cap S^+\big)$ for every $(a, b] \in \calS$; thus $\tilde R'_{a, b} = R'_{a, b} - C\big((a, b] \cap S^+\big)$.  Then the \laminarKC instance we focus on is $(T, C, K, \calS, R')$.  By the definition of $\tilde R'_{a, b}$, we have 
\begin{flalign*} 
\text{either } &&
	\sum_{s \in (a, b] \setminus S^+} \min \set{C_s, \tilde R'_{a, b}} y_s &\geq 2\tilde R'_{a, b}, && \\
\text{or} && \sum_{s \in (a, b] \setminus S^+ : C_s \geq  \tilde R'_{a, b}} y_s &\geq 1. &&
\end{flalign*}

Thus, we can use Theorem~\ref{thm:laminar-main} for the instance $(T, C, K, \calS, R')$ and $y$ to obtain a set $S^* \supseteq S^+$ such that $K(S^*) \leq \sum_{s \in [T]} y_s K_s$ and $C((a', b'] \cap S^*) \geq R'_{a', b'}$ for every set $(a', b'] \in \calS$. 

Now, focus on any interval $(a, b]$ over $[T]$. If $\tilde R_{a, b} = 0$, then $C((a, b] \cap S^+) \geq R_{a, b}$ and the requirement for $(a, b]$ is satisfied; thus we can assume $\tilde R_{a, b} > 0$.  Then by Lemma~\ref{lemma:calS-is-good}, we have a set $[a', b') \in \calS$ such that $[a', b') \subseteq (a, b]$ and $\tilde R'_{a', b'}\geq \tilde R_{a, b}$. Thus, we have that $C\big(S^* \cap (a, b]\big) - C\big(S^* \cap \big((a, a'] \cup (b', b]\big)\big) =   C\big(S^* \cap (a', b']\big)  \geq R'_{a', b'}  = \tilde R'_{a', b'} + C\big((a', b']\cap S^+\big)  \geq \tilde R_{a, b} + C\big((a', b']\cap S^+\big)$. So $C\big(S^* \cap (a, b] \big) \geq \tilde R_{a, b} + C\big((a', b']\cap S^+\big) + C\big(S^* \cap \big((a, a'] \cup (b', b]\big)\big) \geq \tilde R_{a, b} + C\big((a, b] \cap S^+\big) = R_{a, b}$, where the last inequality used the fact that $S^* \supseteq S^+$. This finishes the proof of Theorem~\ref{thm:interval-main}. 
 	\section{Omitted Proofs}
\label{sec:proofs}

\subsection{Proof of Lemma~\ref{lemma:sufficient-condition-for-xstar}}
\begin{proof}
	For every $\newi \in {[N]}$, and for every $s$ from $1$ to $r_\newi$, let $x'_{s, \newi} = \min\set{{\frac52} x_{s, \newi}, 1-x'_{[s-1], \newi}}$. By this definition, we have that $x'_{[t], \newi} = \min\set{{\frac52}x_{[t], \newi}, 1}$ for every $\newi \in {[N]}, t \in [r_\newi]$.  We shall first show that if  $x^*_{[t], \newi} \leq x'_{[t], \newi} = \min\set{{\frac52} x_{[t], \newi},  1}$ for every $\newi \in {[N]}$ and $t \in [r_\newi]$, then $\hcost(x^*) \leq {\frac52}\cdot \hcost(x)$.
	\begin{flalign*}
		&& \hcost(x^*) &= \sum_{\newi \in {[N]}}d_\newi\sum_{s \in [r_\newi]} x^*_{s, \newi}h_\newi(s) = \sum_{\newi \in {[N]}}d_\newi\sum_{s\in [r_\newi]}x^*_{s,\newi}\sum_{t=s}^{r_\newi-1}(h_\newi(t) - h_\newi(t+1)) && \text{since } h_i(r_i) = 0, \forall i \in [N]\\
		&& &= \sum_{\newi \in {[N]}}d_\newi\sum_{t=1}^{r_\newi-1}(h_\newi(t) - h_\newi(t+1))\sum_{s\in [t]}x^*_{s,\newi} &&\\
		&& &\leq {\frac52} \sum_{\newi \in {[N]}}d_\newi\sum_{t=1}^{r_\newi-1}(h_\newi(t) - h_\newi(t+1))\sum_{s\in [t]}x_{s,\newi} = {\frac52}\cdot \hcost(x). &&
	\end{flalign*}
	
	Thus, it suffices to find an $x^* \in [0, 1]^{[T] \times {[N]}}$ such that $\sum_{s \in [T]}x^*_{s,\newi } = 1$ for every $\newi \in {[N]}$, $x^*_{s, \newi} = 0$ if $s \notin S^*$ or $s > r_\newi$,  $\sum_{\newi \in {[N]}}x_{s, \newi}d_\newi \leq C_s$ for every $s \in S^*$, and $x^*_{[t], \newi} \leq x'_{[t], \newi}$ for every $\newi \in {[N]}$ and $t \in [r_\newi]$. This $x^*$ can be found by solving the following fractional $b$-matching instance on the bipartite graph $(A \cup B, E)$, where $A = \set{u_{s, \newi}: \newi \in {[N]}, s \in [r_\newi]}$ and $B = \{v_{s'}: s' \in S^*\}$.   We assign the $d_\newi$ units of demand for item $\newi$ to vertices in $\{u_{s,\newi}: s \in [r_\newi]\}$ according to $x'$: $u_{s, \newi}$ is assigned $x'_{s ,\newi} d_\newi$ units of demand.  In order to guarantee that $x^*_{[t], \newi} \leq x'_{[t], \newi}$ for every $t \in [r_\newi]$, we guarantee that the $x'_{s, \newi}d_\newi$ units of demand assigned to $u_{s, \newi}$ can only be satisfied by orders in $[s, r_\newi]$.  Thus, we define $E$ as follows: for every $u_{s, \newi} \in A$ and $v_{s'} \in B$ such that $s' \in [s, r_\newi]$, we have an edge $(u_{s, \newi}, v_{s'}) \in E$.  The goal of the fractional $b$-matching problem is to find a vector $z \in \R_{\geq 0}^E$ such that $\sum_{e\text{ incident to }u_{s, \newi}}z_e =x'_{s, \newi} d_\newi$ for every $u_{s, \newi} \in A$, and $\sum_{e\text{ incident to }v_{s'}}z_e \leq C_{s'}$ for every $v_{s'} \in B$.  
	
	It is well-known that the above fractional $b$-matching instance is feasible if and only if for every $A' \subseteq A$,  we have that $C(\{s': v_{s'} \in B \text{ is adjacent to some vertex in }A'\}) \geq \sum_{u_{s, \newi} \in A'}x'_{s, \newi}d_\newi$. That is, $C\left(\union_{u_{s, \newi} \in A'}[s, r_\newi]\right) \geq \sum_{u_{s,\newi} \in A'}x'_{s, \newi}d_\newi$.  If $u_{s, \newi}\in A'$, then we can assume that for every $s' \in (s, r_\newi]$, we also have $u_{s', \newi} \in A'$; this does not change the left-side of the inequality but increases the right-side and makes the inequality harder to satisfy. So, we can assume there is a set $I \subseteq {[N]}$, $t \in [0, T)^I$ such that $t_\newi \in [0, r_\newi)$ for every $\newi \in I$ and $A' = \set{u_{s, \newi}:\newi \in I, s \in (t_\newi, r_\newi]}$. Then the neighbors of $A'$ is $\set{v_{s'} :  s' \in \union_{\newi \in I} (t_\newi, r_\newi]}$.  
	
	Thus, to guarantee the existence of $x^*$, it suffices to guarantee that for every such $I$ and $t$, we have 
	\begin{align}
		\sum_{\newi \in I}x'_{(t_\newi, r_\newi], \newi}d_\newi \leq \sum_{s' \in \union_{\newi \in I}(t_\newi, r_\newi]} y^*_{s'} C_{s'}. \label{inequ:sufficient-condition}
	\end{align}
	We can further assume $\union_{\newi \in I}(t_\newi, r_\newi]$ is a time interval; otherwise, we can break $I$ into two sets $I'$ and $I''$ such that $\union_{\newi \in I'}(t_\newi, r_\newi]$ and $\union_{\newi \in I''}(t_\newi, r_\newi]$ are disjoint.  Inequality~\eqref{inequ:sufficient-condition} with $I$ replaced with $I'$, and the inequality with $I$ replaced with $I''$, implies the Inequality~\eqref{inequ:sufficient-condition}.  If $\union_{\newi \in I}(t_\newi, r_\newi] = (a, b]$, then the right side of Inequality~\eqref{inequ:sufficient-condition} is $\sum_{s \in (a, b]}y^*_{s}C_{s}= C(S^* \cap (a, b])$; we want to find the $(I, t)$ pair with $\sum_{\newi \in I}(t_\newi, r_\newi] = (a, b]$ that maximize the left side.   If some $\newi$ has $r_\newi > b$ or $r_\newi \leq a$ then $\newi \notin I$.  Otherwise, we can let $\newi \in I$; and $t_\newi = a$ will maximize $x'_{(t_\newi, r_\newi], \newi}$. So, the maximum possible value of the left side of Inequality~\eqref{inequ:sufficient-condition} is 
	$$\sum_{\newi \in {[N]}: r_\newi \in (a, b]}x'_{(a, r_i], \newi}d_\newi =\sum_{\newi \in {[N]}: r_\newi \in (a, b]}\left(1-x'_{[a], i}\right)d_\newi  = \sum_{\newi \in {[N]}: r_\newi \in (a, b]}\max\set{1-{\frac52}x_{[a], \newi}, 0}d_\newi = R_{a, b}.$$
	Thus, to guarantee the existence of $x^*$, it suffices to guarantee that for every interval $(a, b]$, we have $C(S^* \cap (a, b]) \geq R_{a, b}$.
\end{proof}

\subsection{Proof of Lemma~\ref{lemma:laminar-maintain-feasible}}
\begin{proof}
	We prove the lemma by induction.  Suppose we are at the beginning of first iteration of the outer loop. By the initialization of $S^\circ$ and $S^*$ in Statement~\ref{STATE:laminar-init-S}, and the fact that $y_s = 1$ for every $s \in S^+$,  Constraints~\eqref{LPC:laminar-iteration-Scirc} to \eqref{LPC:laminar-iteration-Sother} are satisfied.  The statement also sets the initial $\hat R$ to be $\tilde R$.  By Inequality~\eqref{inequ:laminar-main-requirement-1} and Inequality~\eqref{inequ:laminar-main-requirement-2}, for every $(a, b] \in \calS$ with $\hat R_{a, b} > 0$, either Constraint~\eqref{LPC:laminar-iteration-calS1} or Constraint~\eqref{LPC:laminar-iteration-calS2} is satisfied. Thus, after Statement~\ref{STATE:laminar-init-calS1}, Constraints~\eqref{LPC:laminar-iteration-calS1} and \eqref{LPC:laminar-iteration-calS2} are satisfied. So, the lemma holds for the first iteration of the outer loop. 
	
	Now, we assume that $y$ is a feasible solution to \ref{LP:laminar-iteration} at the beginning of some iteration of the outer loop; we prove that it is also feasible at the beginning the next iteration, if it exists.  We run Algorithm~\ref{alg:alg-laminar-KC} from the beginning of this iteration.  Since Statements~\ref{STATE:laminar-check-remove} and \ref{STATE:laminar-remove-2} only remove sets from $\calS^1$ and $\calS^2$, they do not destroy the feasibility of $y$. So, $y$ is a feasible solution to \ref{LP:laminar-iteration} before running Statement~\ref{STATE:laminar-solve-lp}. Then, the statement will always find a feasible solution $y$ to \ref{LP:laminar-iteration}. Statement~\ref{STATE:laminar-update-Scirc} only adds knapsacks with $y_s = 0$ to $S^\circ$ and does not destroy the feasibility of $y$.   Lemma~\ref{lemma:middle-loop-maintain-feasibility} says that running an iteration of the middle loop does not destroy the feasibility of $y$.  Thus, $y$ is a feasible solution at the beginning of the next iteration of the outer loop. This finishes the proof of Lemma~\ref{lemma:laminar-maintain-feasible}.	
\end{proof}

\subsection{Proof of Lemma~\ref{lemma:terminate-in-T-iterations}}
\begin{proof}
	We show that in each iteration of the outer loop,  we either have added some new knapsack to $S^\circ$ in Statement~\ref{STATE:laminar-update-Scirc}, or have added some new knapsack to $S^*$ in Statement~\ref{STATE:laminar-update-Sstar}. This proves that the algorithm will terminate in at most $T$ iterations since there are only $T$ knapsacks and $S^\circ \cap S^* = \emptyset$. 
	
	Let us run the algorithm from the beginning of an iteration of the outer loop, at which time we have $\calS^1 \cup \calS^2 \neq \emptyset$. Statements~\ref{STATE:laminar-check-remove} and \ref{STATE:laminar-remove-2} remove a set from $\calS^1 \cup \calS^2$ only if $\calS^1 \cup \calS^2$ contains at least two sets. Thus after Statement~\ref{STATE:laminar-remove-2} we still have $\calS^1 \cup \calS^2 \neq \emptyset$. So before  Statement~\ref{STATE:laminar-solve-lp}, there must be some constraint of form \eqref{LPC:laminar-iteration-calS1} or \eqref{LPC:laminar-iteration-calS2} in the LP.  Thus, $S^\circ \cup S^* \neq {[T]}$ since otherwise $y$ can not satisfy the constraint due to Observations~\ref{obs:Sstar-Scirc} and \ref{obs:hatR}, contradicting Lemma~\ref{lemma:laminar-maintain-feasible}.
	
	Then we run Statement~\ref{STATE:laminar-solve-lp} to find a vertex-point solution $y$ to \ref{LP:laminar-iteration}.  We assume towards the contradiction that we have $S^\circ = \set{s \in {[T]}: y_s = 0}$ and $S^* = \set{s \in {[T]}: y_s = 1}$; otherwise we will add some knapsack to $S^\circ$ or $S^*$ later.  We choose a set of $T$ linearly independent tight constraints that defines $y$.  We require that this set contains all tight constraints of the form~\eqref{LPC:laminar-iteration-Scirc}, \eqref{LPC:laminar-iteration-Sstar} and \eqref{LPC:laminar-iteration-Sother} (these tight constraints are linearly independent).  The number of tight constraints of form~\eqref{LPC:laminar-iteration-calS1} and \eqref{LPC:laminar-iteration-calS2} in the independent set is exactly $T - |S^\circ| - |S^*| = \big|{[T]} \setminus (S^\circ \cup S^*)\big| \geq 1$.  Let $\calS' \subseteq \calS^1 \cup \calS^2$ be the family of intervals that corresponds to these tight constraints; so $|\calS'| = \big|{[T]} \setminus (S^\circ \cup S^*)\big| \geq 1$.  
	
	If we have two sets $(a, b], (a', b'] \in \calS' \subseteq \calS^1 \cup \calS^2$ such that $(a', b'] \subsetneq (a, b]$, then Statements~\ref{STATE:laminar-check-remove} and \ref{STATE:laminar-remove-2} guaranteed that $(a', b'] \setminus (S^\circ \cup S^*) \subsetneq (a', b'] \setminus (S^\circ \cup S^*)$.  Otherwise $(a', b'] \setminus (S^\circ \cup S^*) = (a', b'] \setminus (S^\circ \cup S^*)$ and Statements~\ref{STATE:laminar-check-remove} and \ref{STATE:laminar-remove-2} must have removed either $(a, b]$ or $(a', b']$ from $\calS^1 \cup \calS^2$.  Thus, there must be a knapsack in $(a, b] \setminus (a', b']$ that is not in $S^\circ \cup S^*$.  
	
	Now we focus on the laminar forest defined by the set $\calS'$ (the forest is not empty). We assign each knapsack $s \in {[T]} \setminus (S^\circ \cup S^*)$ to the minimal set $(a, b] \in \calS'$ that contains $s$. If some $(a, b] \in \calS'$ has exactly one child $(a', b']$ in the laminar forest, then $((a, b] \setminus (a', b']) \setminus (S^\circ \cup S^*) \neq \emptyset$ and some knapsack must be assigned to $(a, b]$.  The number of leaves in the laminar forest is strictly more than the number of non-leaves that have at least two children. Since the number of nodes in the laminar forest is $|\calS'| = \big|{[T]} \setminus (S^\circ \cup S^*)\big|$, and each inner node with exactly one child is assigned at least one knapsack in ${[T]} \setminus (S^\circ \cup S^*)$, there must be a leaf $(a, b]$ in the forest such that $(a, b]$ is assigned at most one knapsack in ${[T]} \setminus (S^\circ \cup S^*)$. For this $(a, b]$, we have $\sum_{s \in (a, b] \setminus S^*} y_s = \sum_{s \in (a, b] \setminus (S^\circ \cup S^*)} y_s + \sum_{s \in (a, b] \cap S^\circ} y_s  < 1$, since $|(a, b] \setminus (S^\circ \cup S^*)| \leq 1$, every  $s \notin S^*$ has $y_s < 1$, and every $s \in (a, b] \cap S^\circ$ has $y_s = 0$.  Recall that $\hat R_{a, b} > 0$ by Observation~\ref{obs:hatR}.  If $(a, b] \in \calS^1$, then $\sum_{s \in (a, b] \setminus S^*} \min\set{C_s, \hat R_{a, b}} y_s < \sum_{s \in (a, b] \setminus S^*} \hat R_{a, b} y_s   < \hat R_{a, b} < 2\hat R_{a, b}$, contradicting Constraint~\eqref{LPC:laminar-iteration-calS1} for $(a, b]$. If $(a, b] \in \calS^2$, then $\sum_{s \in (a, b] \setminus S^*} y_s  < 1$ contradicts Constraint~\eqref{LPC:laminar-iteration-calS2} for $(a, b]$. This finishes the proof of the lemma.
\end{proof}

\end{document}